\documentclass[12pt]{article}

\usepackage{sbc-template}
\usepackage{graphicx,url}
\usepackage[utf8]{inputenc}
\usepackage[english]{babel}
\usepackage[linesnumbered,ruled,vlined]{algorithm2e}
\usepackage[shortlabels]{enumitem}
\usepackage{amsthm}

\usepackage{xcolor}
\usepackage{fancyhdr}
\newtheorem{definition}{Definition}
\newtheorem{theorem}{Theorem}
\newtheorem{lemma}[theorem]{Lemma}
\newtheorem{problem}{Problem}

\definecolor{gray160}{RGB}{160,160,160} 
\fancypagestyle{firstpage}{
  \fancyhf{} 
  \fancyhead[C]{
    \begin{center}
      \fontsize{10}{6}\selectfont\textcolor{gray160}{
      This article has been published in the Proceedings of XXIV GEOINFO, after peer review. \\
      DBLP Entry: https://dblp.org/rec/conf/geoinfo/MinervinoCOS23.html \\
    Conference Repository: http://urlib.net/ibi/8JMKD3MGPDW34P/4ADBK2H 
    }
    \end{center}
    }
} 

\title{QQESPM: A Quantitative and Qualitative Spatial Pattern Matching Algorithm}

\author{Carlos V. A. Minervino\inst{1}, Cláudio E. C. Campelo\inst{1},\\ Maxwell Guimarães de Oliveira\inst{1}, Salatiel D. Silva\inst{1} }

\address{Systems and Computing Department (DSC) \\ Federal University of Campina Grande (UFCG) \\ Av. Aprígio Veloso 882, Bloco CN, Bairro Universitário -- 58.429-140 \\ Campina Grande -- PB -- Brazil}

\begin{document} 

\maketitle


\thispagestyle{firstpage}

\begin{abstract}
   The Spatial Pattern Matching (SPM) query allows for the retrieval of Points of Interest (POIs) based on spatial patterns defined by keywords and distance criteria. However, it does not consider the connectivity between POIs. In this study, we introduce the Qualitative and Quantitative Spatial Pattern Matching (QQ-SPM) query, an extension of the SPM query that incorporates qualitative connectivity constraints. To answer the proposed query type, we propose the QQESPM algorithm, which adapts the state-of-the-art ESPM algorithm to handle connectivity constraints. Performance tests comparing QQESPM to a baseline approach demonstrate QQESPM's superiority in addressing the proposed query type.

\end{abstract}

\section{Introduction}
\label{sec:introduction}
\protect

The rise of Location-Based Services (LBS) such as Google Maps\footnote{https://www.google.com/maps} and Foursquare\footnote{https://foursquare.com/} has underscored the necessity for efficient Points of Interest (POIs) search algorithms. The continuous expansion of geotextual data within these systems outlines the importance of effective algorithms and mechanisms for efficient POI querying based on attributes such as keywords, proximity, and other factors.

Spatial Pattern Matching (SPM), a category of POI group search, is designed to identify all combinations of POIs that conform to a user-defined spatial pattern established by keywords and distance criteria \cite{msj2018, fang2019evaluating, fang2018spacekey, li2019spatial, espm2019}. To illustrate, consider a scenario where a user seeks an apartment near a school and a hospital, while maintaining a certain distance from the hospital for hygiene reasons. The user's criteria stipulate that the apartment should be situated between 200m and 1km away from a hospital and at most 2km away from a school. Such requirements can be addressed through an SPM search by using the query pattern depicted in Figure \ref{fig:examplespm} (A).

While the SPM search methodology proves highly effective in scenarios necessitating distance constraints among queried POIs, it lacks the capability to address qualitative connectivity constraints between these entities. For instance, it cannot handle queries such as finding a school adjacent to a wooded area. To illustrate a more intricate search scenario, consider an individual seeking a rental space within a commercial building for their small business. In addition to this, they require the building to have an onsite gym and an adjacent green area. Furthermore, they need the building to be located within 1km from an elementary school for the convenience of the enrollment and pickup of their child. This complex scenario can be modeled using a spatial pattern graph that incorporates both quantitative (distance) and qualitative (connectivity) constraints, as shown in Figure \ref{fig:examplespm} (B). However, existing SPM search algorithms are unable to handle such queries, requiring users to perform distance-based queries and manually sift through the results to find those meeting qualitative constraints.
\protect
\protect
\begin{figure}[ht]
\begin{center}
\includegraphics[width=.9\textwidth]{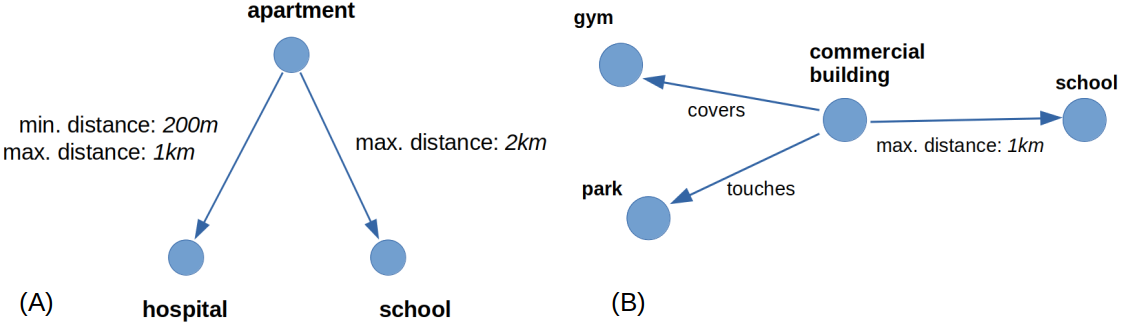}
\caption{Example of a distance-based spatial pattern (A) and a qualitative and quantitative spatial pattern (B)}
\label{fig:examplespm}
\end{center}
\end{figure}
\protect
Considering this challenge, this paper introduces a new type of POI group search: Qualitative and Quantitative Spatial Pattern Matching (QQ-SPM). The QQ-SPM query extends the conventional SPM query to encompass qualitative connectivity constraints between queried POIs. This approach enables the incorporation of qualitative requirements expressed through topological connectivity relations among the POI geometries. QQ-SPM thus provides a comprehensive solution that covers the entire spectrum of SPM search patterns while accommodating qualitative criteria specified by the user, enhancing the versatility of spatial pattern specification.
\protect
This work has the following key contributions:
\protect\begin{itemize}\protect
\protect
\item\protect A Formal Definition of the QQ-SPM query, where the central parameter is a spatial pattern represented as a graph. This pattern defines the target POI keywords, desired distances, and connectivity relations.
\item The QQESPM algorithm, designed to address QQ-SPM queries. This algorithm is adapted from the Efficient Spatial Pattern Matching (ESPM) algorithm presented in \cite{espm2019}, specifically refined to accommodate connectivity requirements.
\item An open-source code implementation for the proposed QQESPM algorithm.
\item An Empirical Evaluation with comparative analyses, to assess the efficiency and scalability of the QQESPM algorithm. This evaluation compares the performance of QQESPM with that of a basic solution that employs qualitative constraint verification only during the final stage of a traditional SPM query.
\end{itemize}
\protect
The rest of the paper is organized as follows. Section 2 summarizes related work. Section 3 brings a review of the  indexing and the topological relation model used in the QQESPM algorithm. Section 4 brings the formal definition of the concepts that permeate the QQ-SPM query problem. Section 5 describes the proposed QQESPM algorithm. Section 6 outlines the performance experiments comparing the proposed QQESPM algorithm with a trivial solution. Finally, Section 7 concludes the paper by summarizing the main achievements.
\protect
\section{Related Work} \label{sec:relatedwork}
\protect
In this section we mention three of the main types of spatial keyword queries related to this work. The first type involves searching for POIs that meet specific keywords and are in close proximity to a designated center point for the search. For instance, the top-k spatial keyword search aims to identify geotextual objects (e.g., POIs) using a set of keywords and an initial search location. The goal is to locate the top-k closest POIs to the starting point while satisfying all search keywords. Studies in this field include those by \cite{cao2011collective, hermoso2019re, zhang2013scalable}.
\protect
The second type of search focuses on minimizing distances between queried POIs. For instance, m-Closest spatial keyword search seeks groups of closely located POIs that collectively satisfy a user-defined set of $m$ keywords. Studies in this category include those conducted by \cite{choi2016finding, choi2020spatial, guo2015efficient}. However, the first two search types lack the capability to accommodate more intricate patterns, such as specifying a minimum distance between two returned POIs, which is essential when users want to avoid close proximity to certain types of POIs, like hospitals.
\protect
The third search type is the SPM search, which utilizes a graph-based spatial pattern. In this pattern, vertices store spatial keywords, and edges represent desired distance constraints. SPM search offers increased specificity by enabling users to impose both minimum and maximum distance restrictions between pairs of POIs. Studies in this category include works by \cite{msj2018, li2019spatial, fang2018spacekey, fang2019evaluating, espm2019, chen2022example}. However, the SPM search lacks the capability to model qualitative restrictions, such as connectivity limitations.
\protect
In \cite{long2016indexing}, an efficient mechanism for indexing qualitative relations is proposed, aiming to reduce the time required for calculating the qualitative relation between two geometries. The core concept involves initially computing the qualitative relation between the Minimal Bounding Rectangles (MBRs) of the spatial objects. In cases where it is not possible to determine the topological relation between the geometries of the POIs solely based on the topological relation between their MBRs, their topological relation will be previously indexed. However, this approach primarily focuses on efficiently determining the qualitative relation between two existing geospatial objects within a dataset, rather than identifying the subset of objects satisfying a specific qualitative relation among numerous objects.
\protect
The concept of a spatial pattern defined by qualitative connectivity constraints is introduced in \cite{gabriel2021master}. The work presents the Qualitative Spatial Pattern Search (QSPM) and an algorithm called the Topo-MSJ algorithm for addressing this type of search. However, being an early work in the realm of qualitative patterns, the author does not explore the potential of combining quantitative restrictions with qualitative ones in this search context.
\protect
\section{Background}
\protect
In this Section we give a brief review of the core concepts used in the QQESPM algorithm, including the geo-textual indexing and the topological relation model.
\protect
\subsection{IL-Quadtree}
\label{subsec:ilquadtree}
\protect
The Inverted Linear Quadtree (IL-Quadtree), a geotextual indexing structure introduced in \cite{ilquadtree2016}, serves as a fundamental component in the QQESPM algorithm. This index functions as a map, associating each unique spatial keyword in the dataset with its respective quadtree index structure. Each of these quadtrees contains a set of spatial objects (e.g., POIs) related to the specific keyword under consideration.
\protect
The bidimensional quadtree, as proposed in \cite{finkel1974quad}, divides a two-dimensional spatial domain into four quadrants recursively. Each quadrant can be further subdivided into four subquadrants, and this subdivision is represented by a tree structure. Each rectangular subspace represents a node in the tree, and a node's children correspond to its subquadrants. Subdivision occurs when the number of spatial entities (POIs) in a node exceeds a specified threshold, which can be adjusted during quadtree construction. Subspace division employs directional codes (00, 01, 10, and 11) to signify southwest, southeast, northwest, and northeast quadrants, respectively. Concatenating these codes recursively provides a unique identifier for each node, indicating its position in the quadtree hierarchy. Figure \ref{fig:quadtree} illustrates the spatial partitioning in the quadtree and its associated tree structure. The IL-Quadtree's architecture efficiently retrieves geotextual objects during geotextual queries, as indicated in \cite{ilquadtree2016} and \cite{espm2019}. The QQESPM algorithm relies on the IL-Quadtree indexing method to perform queries. 
\protect
\begin{figure}[ht]
\centering
\includegraphics[width=.8\textwidth]{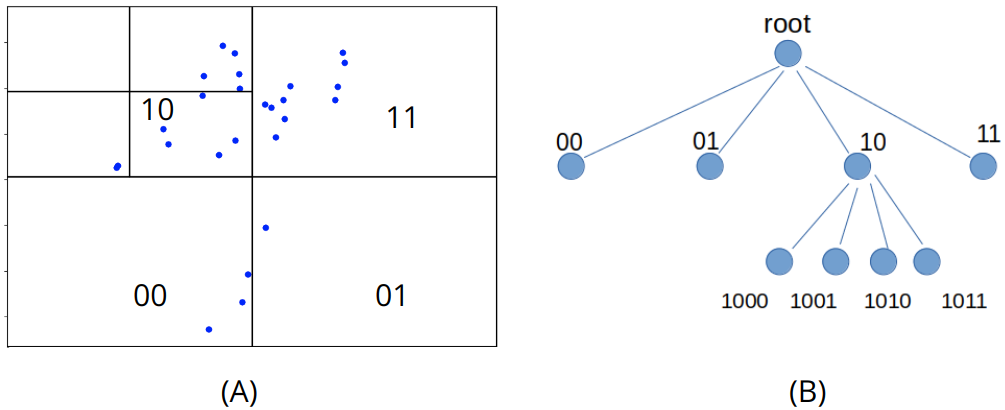}
\caption{Example of a quadtree space subdivision (A), and its associated tree structure (B)}
\label{fig:quadtree}
\end{figure}
\protect
\subsection{DE-9IM}
\label{subsec:de9im}
\protect
A foundational model for computing the topological connectivity relation between two-dimensional geometries is the Dimensionally Extended Nine-Intersection Model (DE-9IM) \cite{egenhofer1990categorizing, clementini1993small, clementini1994modelling}. This representation provides a structured framework for formally defining spatial predicates that describe the connectivity between POIs. DE-9IM can represent topological relations such as ``\textit{equals}", ``\textit{touches}" and  ``\textit{contains}".
\protect
This topological relation model utilizes a 3x3 matrix to represent the topological relation between two distinct geometries, denoted as A and B. The matrix elements represent intersections across the interior, boundary, and exterior components of these geometries. Each matrix configuration corresponds to a possible topological relation. A simple description for some topological relations from \cite{strobl2008dimensionally} can be found in Table \ref{tab:topologicalpredicates}.
\protect
\begin{table}[h]
\caption{Topological Predicates} 
\protect
\label{tab:topologicalpredicates}
 \small
  \begin{center}
    \begin{tabular}{|c||p{11cm}|} \hline
      {\bf Topological Predicate}   &{\bf Meaning} \\\hline
      {Equals}   & The Geometries are topologically equal \\\hline
      {Disjoint}   & The Geometries have no point in common \\\hline
      {Intersects}   & The Geometries have at least one point in common (the inverse of Disjoint) \\\hline
      {Touches}   & The Geometries have at least one boundary point in common, but no interior points\\\hline
      {Partially Overlaps}   & The Geometries share some but not all points in common, and the intersection has the same dimension as the Geometries themselves\\\hline
      {Within}   & Geometry A lies in the interior of Geometry B\\\hline
      {Contains}  & Geometry B lies in the interior of Geometry A (the inverse of Within)\\\hline
    \end{tabular}
  \end{center}
\end{table}
\protect
The proposed QQESPM algorithm uses the topological relations ``\textit{equals}", ``\textit{touches}", ``\textit{covers}",  ``\textit{covered by}", ``\textit{partially overlaps}" and ``\textit{disjoint}". The relation ``\textit{covers}" is a variation of ``\textit{contains}" allowing the geometries to have intersecting boundaries \cite{clementini1994modelling}, and the relation ``\textit{covered by}" is simply the inverse of ``\textit{covers}".
\protect
\section{Problem Formalization}
\protect
Within this section, we give a formal definition for the fundamental terms in the QQ-SPM search problem.
~\newline
\begin{definition}[spatial pattern]
A Spatial Pattern is a graph $G(V, E)$ with a set of $n$ vertices $V = {v_1, … v_n}$ and a set of $m$ edges $E$, satisfying the following constraints:
\protect
\begin{enumerate}[(a)]
\item each vertex $v_i\in V$ has an associated spatial keyword $w_i$
\item each edge $e(v_i, v_j)\in E$ is labelled with at least one of the following types of description:
    \begin{itemize}
    \item[\textbf{-}] a conectivity spatial predicate $\Re_{ij}$, among the following: $\{``\textit{equals}", ``\textit{touches}", ``\textit{covers}", ``\textit{covered\ by}", ``\textit{partially\ overlaps}",\newline ``\textit{disjoint}"\}$ 
    \item[\textbf{-}] a distance interval $[l_{ij}, u_{ij}]$, and a sign $\tau\in\{ ``\leftarrow ", ``\rightarrow ", ``\leftrightarrow ", ``- " \}$\ 
    \end{itemize}
\end{enumerate}
\end{definition}
\protect
Each possible spatial pattern graph specifies a QQ-SPM query. The vertices specify the POIs desired keywords. The connectivity predicate indicates the desired connectivity relation between the searched POIs. The distance between the searched POIs is restricted by the lower ($l_{ij}$) and upper ($u_{ij}$) bounds associated with the edge. The meanings of the possible signs for an edge are described below:
\begin{itemize}
\item $v_i\rightarrow v_j$ [$v_i$ excludes $v_j$ ]: No POI with keyword $w_j$ in the dataset should be found within a distance less than $l_{ij}$ from the POI corresponding to $v_i$.
\item $v_i\leftarrow v_j$ [$v_j$ excludes $v_i$]: No POI with keyword $w_i$ in the dataset should be found within a distance less than $l_{ij}$ from the POI corresponding to $v_j$.
\item $v_i\leftrightarrow v_j$ [mutual exclusion]: The two-way restriction, i.e., $v_i$ excludes $v_j$ and $v_j$ excludes $v_i$.
\item $v_i\ –\ v_j$ [mutual inclusion]: The occurrence of POIs with keywords $w_i$ and $w_j$ in the dataset with distance shorter than $l_{ij}$ from POIs corresponding to $v_i, v_j$ is allowed.
\end{itemize}
\protect
Edges with the distance interval information are called quantitative edges, and edges with the connectivity predicate are called qualitative edges. Edges may or may not be simultaneously quantitave and qualitative. If a quantitative edge has the mutual inclusion sign, it is called an inclusive edge, otherwise, it is called an exclusive edge. Also note that, since the relation ``\textit{covered by}" is the inverse of ``\textit{covers}", it could be discarded, but once edges are directional, i.e., have specific starting and ending vertices, we keep the relation ``\textit{covered by}".
\protect
Notice that the attributes of an edge for the QQ-SPM query is a generalization of the attributes of an edge for the SPM query allowing qualitative connectivity constraints. In this way, the spatial pattern definition for the QQ-SPM query is also a generalization of the spatial pattern definition for the SPM query.
~\newline
\begin{definition}[qq-e-match] A pair of POIs $(p_i, p_j)$ from a dataset $D$ is called a qq-e-match for the edge $e(v_i, v_j)$ if they respectively have the keywords $w_i, w_j$ from the vertices $v_i, v_j$, and satisfy the distance and connectivity constraints of the edge $e$. 
\end{definition}
~\newline
\begin{definition}[match]
A tuple of $n$ POIs $S = (p_1, p_2, …, p_n)$ from a dataset $D$ is called a match for a spatial pattern $G(V, E)$ when $|V|=n$ and for each $1\le i\le n$, $p_i$ has the keyword $w_i$ from the vertex $v_i$, and for each edge $e(v_i, v_j)$ of $G$, the POIs pair $(p_i, p_j)$ is a qq-e-match for the edge $e$. 
\end{definition}
Note that the order of POIs in the tuple corresponds to the order of vertices in the pattern $G$, so the $i$th POI $p_i$ in the tuple corresponds to the $i$th vertex ($v_i$) in the pattern $G$.
~\newline
\begin{problem}[QQ-SPM query]
The QQ-SPM search problem or QQ-SPM query consists of finding all the matches of a spatial pattern $G(V, E)$ in a dataset $D$ of POIs, i.e., finding all combinations of POIs from $D$ that match the given spatial pattern.
\end{problem}
\protect
In order to calculate the qq-e-matches efficiently, the QQESPM algorithm uses the qq-n-match concept, formally defined below.
~\newline
\begin{definition}[qq-n-match]
\label{def:nmatch}
Let $e(v_i, v_j)$ be an edge of a spatial pattern $G(V, E)$, let $ILQ_i$ and $ILQ_j$ be the quadtrees for the keywords $w_i$ and $w_j$ of the vertices $v_i$ and $v_j$ respectively, and let $n_i, n_j$ be two nodes from $ILQ_i$ and $ILQ_j$, and $b_i, b_j$ the MBRs for the nodes $n_i, n_j$ respectively. We say that the node pair $(n_i, n_j)$ is a qq-n-match for the edge $e(v_i, v_j)$ if $d_{min}(b_i, b_j) \le u_{ij}$ and $d_{max}(b_i, b_j) \ge l_{ij}$, where $d_{min}$ and $d_{max}$ represent the minimum and maximum distance between the MBRs, and additionally:
\begin{enumerate}[(a)]
\item Case $v_i\rightarrow v_j$: $\neg \exists n'_j \in ILQ_j$  such that $n'_j\neq n_j\land d_{max}(b_i, b'_j ) < l_{ij}$
\item Case $v_i\leftarrow v_j$: $\neg \exists n'_i \in ILQ_i$ such that $n'_i\neq n_i\land d_{max}(b'_i, b_j ) < l_{ij}$
\item Case $v_i\leftrightarrow v_j$: (a) and (b) holds
\item Case $e$ is qualitative with $\Re_{ij}\neq ``disjoint"$: $b_i\cap b_j\neq\emptyset $ 
\end{enumerate}
\end{definition}
\protect
Intuitively, a pair of nodes $n_i, n_j$ is a qq-n-match for the edge $e$ when by checking the minimum and maximum distance between their MBRs $b_i, b_j$, it is not possible to eliminate the possibility of existing POIs $p_i, p_j$ inside these nodes, such that $(p_i, p_j)$ is a qq-e-match for the edge $e$, so the children nodes or leaves of $n_i, n_j$ need to be further examined, and are candidates for finding qq-e-matches for the edge $e$ in context. Next, we introduce a lemma on which the QQESPM algorithm is based.
\protect
~ \newline
\begin{lemma}
Suppose the node pair $(n_i, n_j)$ is a qq-n-match of the edge $e(v_i, v_j)$. Let $n^f_i$ and $n^f_j$ be the father nodes of $n_i$ and $n_j$ respectively. Then, the node pair $(n^f_i, n^f_j)$ is also a qq-n-match of $(v_i, v_j)$.
\end{lemma}
\protect
\begin{proof}
The proof for the quantitative restrictions of the edges is provided in \cite{espm2019}. Regarding the additional proposed criterion to qualify as a qq-n-match, which is related to the connectivity constraint of the edge, it's important to note that if the node pair $(n_i, n_j)$ constitutes a qq-n-match for an edge with a qualitative constraint other than \textit{disjoint}, they will possess intersecting bounding boxes. Since their father nodes encompass them, they too will be intersecting, thus ensuring that the condition for a node match persists for the father nodes.
\protect
\end{proof}
\protect
\section{QQESPM algorithm}
\protect
\begin{algorithm}[b!]
\small
\DontPrintSemicolon 
\KwIn{IL-Quadtree $ILQ$, spatial pattern $G$}
\KwOut{$\psi$: all the matches of $P$}
$L = max(depth(ILQ_{i}), 1 \le i \le n) $\;
\For{$level = 1$ to $L$}{
    derive the order of computing qq-n-matches for this level\;
        
    \For{each edge $e$}{
            compute the qq-n-matches for $e$ in the current level\;
    }
}
derive the order of computing qq-e-matches\;
identify skip-edges\;
\For{each non-skip\ edge $e$}{
            compute the qq-e-matches for e\;
}
derive the order of joining qq-e-matches\;
$\psi \gets join\_qq\_e\_matches()$\;
\Return{$\psi$}\;
\caption{{\sc QQESPM}}
\label{algo:qqespm}
\end{algorithm}
\protect
This section presents the QQESPM algorithm, designed to handle QQ-SPM queries. QQESPM considers six possible topological relations between POIs, namely ``\textit{equals}", ``\textit{touches}", ``\textit{covers}", ``\textit{covered\ by}", ``\textit{partially\ overlaps}", and ``\textit{disjoint}". The overview of the search procedure is shown in Algorithm \ref{algo:qqespm} (QQESPM), which delineates the high-level sequential steps for query execution, with an emphasis on achieving efficient execution by using the qq-n-matches concept. It iteratively operates at the depth levels in the  quadtrees of keywords in the search pattern, examining qq-n-matches for each edge at the current depth level by evaluating child node pairs from the previous depth level (according to Lemma 1). Upon reaching the final depth level of the quadtrees, it tests the pairs of objects within the last level's qq-n-matches to identify qq-e-matches. The qq-e-matches of each edge are then joined to produce solutions (matches) for the spatial pattern. 
\protect
Within the context of a given edge $e(v_i, v_j)$, the QQESPM algorithm systematically searches for its qq-n-matches in a tiered manner, starting from the root nodes of the $ILQ_i$ and $ILQ_j$ quadtrees. The initial qq-n-match for edge $e$ arises from the pair of root nodes, specifically (root of $ILQ_i$, root of $ILQ_j$), exclusively at level 0. The process then progresses to examine pairs of nodes $(n_i, n_j)$, where $n_i$ is a child of the root of $ILQ_i$ and $n_j$ is a child of the root of $ILQ_j$, following Lemma 1. This step identifies the qq-n-matches for $e$ at the 1st level. The algorithm continues this exploration iteratively by inspecting the children of these nodes, deducing qq-n-matches for $e$ at the 2nd level and subsequent levels. This iterative traversal persists until the maximum depth of the quadtrees is reached. The final level's qq-n-matches are retained to subsequently derive the qq-e-matches for the edges.
\protect
At each level, the algorithm employs a reordering strategy for the edges list, giving priority to edges with fewer qq-n-matches from the previous level, guided by the rationale that such edges are more likely to yield fewer qq-e-matches. This strategic ordering accelerates computation by swiftly eliminating unsuitable nodes early, as proposed by \cite{espm2019}.
\protect
After calculating qq-n-matches for all edges at the final level (the maximum depth of the quadtrees), the algorithm evaluates POI pairs within each edge's qq-n-matches to determine the qq-e-matches. Before computing qq-e-matches for any edge, the algorithm checks if the terminal vertices of the edge have a restricted set of candidate objects. This set is obtained from qq-e-matches of edges with shared vertices. This strategy serves as a pre-joining mechanism avoiding redundant pair assessments. Also, the calculation of qq-e-matches is not necessary for some mutually inclusive edges whose extreme vertices are shared with other edges whose qq-e-matches will be computed, so that for these edges, called skip edges, the verification of constraints occurs at the joining stage, which compares the qq-e-matches of edges with shared vertices, and eliminates the non-matching ones.
\protect
The structural framework and strategic heuristics of ESPM are replicated within the QQESPM algorithm. The divergence lies in the criteria for qq-n-match and qq-e-match identification, as defined in Definitions 2 and 4, which will occur in lines 5 and 9. These divergent criteria is sufficient to promote distinct computations at each level, as qq-n-matches are computed level by level from the root to the maximum depth level of the keyword's quadtrees.
\protect
\section{Experiments and Results}
\protect
In this Section, we evaluate the performance of the proposed algorithm QQESPM in terms of execution time by comparing with a trivial algorithm for solving the QQ-SPM query that we call QQ-simple. 
\protect
\subsection{Experiments Description}
\protect
The ESPM algorithm was implemented in Python, following the description provided in \cite{espm2019}. This implementation was further adapted to include the qq-n-matches and qq-e-matches verification stages to accommodate qualitative connectivity constraints, resulting in the initial implementation of the QQESPM algorithm\footnote{The code for this implementation can be found in https://zenodo.org/records/10085300}. Additionally, a more straightforward approach, referred to as QQ-simple, was also implemented to be compared with QQESPM. This approach checks the connectivity constraints only at the final step of ESPM, employing a filtering mechanism to exclude solutions that do not meet the connectivity requirements. Subsequently, we conducted a comparative performance analysis of these two QQ-SPM solutions.
\protect
Experiments  were executed on a machine equipped with Intel Core i7-12700F CPU 4.90 GHz, coupled with 32GB of memory, operating on the Ubuntu OS. The computational load was carried out on a single CPU core.
\protect
We used a dataset comprising 33,877 POIs, extracted from OpenStreetMap\footnote{https://www.openstreetmap.org/} filtered by the following bounding box: \{min\_lat: -8.3610, min\_long:-38.8559, max\_lat: -5.9275, max\_long: -34.7415\}, thereby predominantly spanning the Paraíba state, Brazil. The dataset comprises the tags `\textit{amenity}', `\textit{shop}', and `\textit{tourism}', containing 315 distinct keywords.
\protect
In an effort to construct resource-intensive search spatial patterns, mirroring real-world conditions, the 20 most frequent keywords were identified and selected from the tags `\textit{amenity}', `\textit{shop}', and `\textit{tourism}', amounting to a cumulative set of 60 keywords to compose the search patterns.  
\protect
The experiments encompassed 12 distinct graph architectures for spatial patterns, following \cite{espm2019}. These structural patterns can be visualized in Figure \ref{fig:structures}. For each of these architectures, 5 distinct spatial patterns were generated with randomly selected keywords, totalizing 60 spatial patterns. 
\protect
The dataset was randomly shuffled and divided into segments representing $20\%$, $40\%$, $60\%$, $80\%$, and $100\%$ of the total POI count. For each of these dataset subsets, searches were conducted five times for each of the 60 spatial patterns generated, for both algorithms QQESPM and QQ-simple. Thus, the total number of executions was 3,000, and each of the two algorithms answered 1,500 queries.
\protect
For the purpose of this study, a simplified convention was adopted, by using the Euclidean distance to measure the distance between POIs coordinates (longitude and latitude). Note that it differs from the distance in kilometers. To construct the search patterns, the parameter $l_{ij}$, representing the minimum inter-POI distance, was randomly drawn from the interval $[0, 0.005]$ (equivalent to approximately $0$ to $550m$), while the parameter $u_{ij}$, representing the maximum inter-POI distance, was randomized within the range $[l_{ij}, l_{ij} + 0.02]$ (reaching up to $2.7km$ approximately). The connectivity relations were introduced in the edges randomly from the set \{``\textit{equals}" , ``\textit{touches}", ``\textit{covers}", ``\textit{covered\ by}", ``\textit{partially\ overlaps}", ``\textit{disjoint}"\}. Each edge had a probability of $50\%$ of receiving a connectivity relation constraint.
\protect
\begin{figure}[t!]
\centering
\includegraphics[width=.5\textwidth]{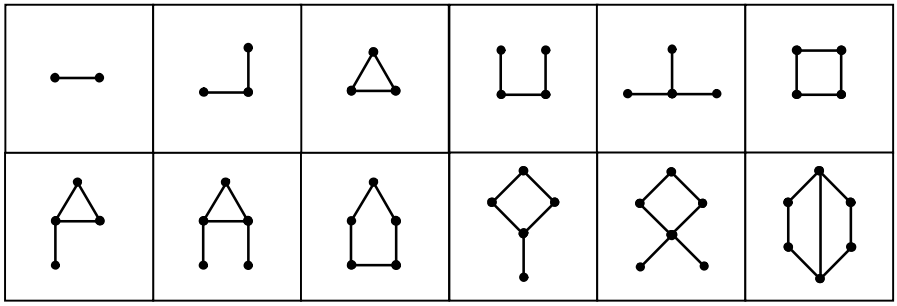}
\caption{Structure of Search Spatial Patterns \cite{espm2019}}
\label{fig:structures}
\end{figure}
\protect
\begin{figure}[t!]
\centering
\includegraphics[width=.9\textwidth]{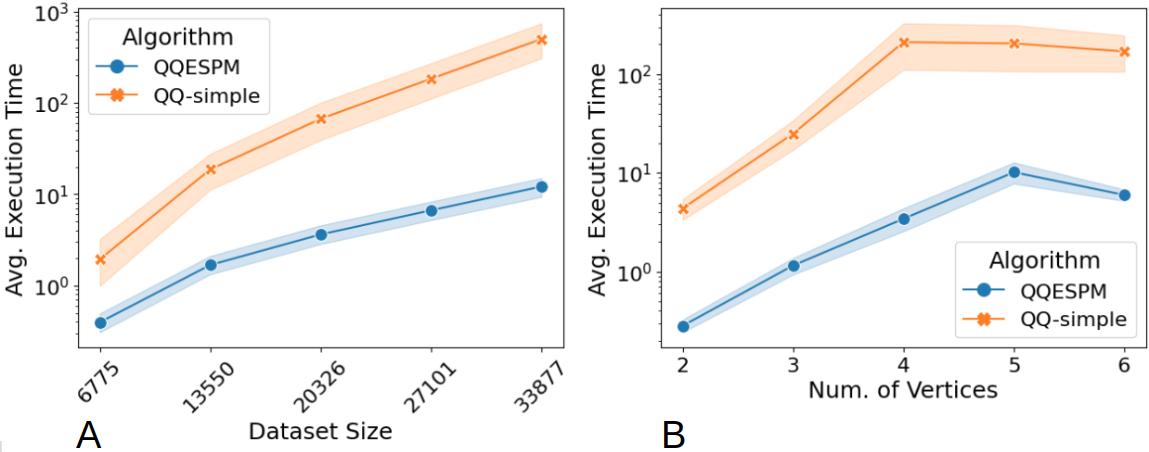}
\caption{Avg. Execution Time by Dataset Size (A) and by Number of Vertices (B) for Algorithms QQESPM and QQ-simple}
\label{fig:results}
\end{figure}
\protect
\subsection{Results}
\protect
We now present the performance results of the executions in terms of scalability of dataset size and variation in the number of vertices.
\protect
\vspace{-0.3cm}
\protect
\subsubsection*{Scalability Assessment}
 The average execution time by dataset size was measured for each algorithm. The results shown in Figure \ref{fig:results} (A) reveal that the average execution time difference between QQESM and QQ-simple becomes larger as the dataset size increases. Clearly, QQESPM demonstrates significantly better scalability compared to QQ-simple. To better visualize the comparison, we applied a logarithmic scale to the y-axis. Shaded areas represent a 95\% confidence interval for the average execution time.
 \protect
\vspace{-0.3cm}
\protect
\subsubsection*{Number of Vertices Assessment}
The average execution time by number of vertices in the search pattern was measured for each algorithm. The results, illustrated in Figure \ref{fig:results} (B), consistently demonstrate QQESPM's superior runtime performance over the basic QQ-simple solution, regardless of the number of vertices. Notably, the observation that patterns with 5 or 6 vertices do not require longer execution times than patterns with 4 vertices can be explained by the increased complexity of search patterns, since in the same search area, patterns with too much keywords are less likely to have matches, and in these cases an early stopping of the query procedure can occur by identifying the non-existence of qq-n-matches in an early level.
\protect
The average memory allocation by QQESPM queries was also consistently lower for all dataset sizes and number of vertices evaluated, compared to the QQ-simple executions. The overall average memory allocation during queries was $284.16$ MB for QQESPM executions and $314.34$ MB for QQ-simple executions, highlighting the memory efficiency advantage of QQESPM over the QQ-simple trivial approach.
\protect
\vspace{-0.3cm}
\protect
\subsubsection*{Statistical Test}
The executions that used the whole dataset were grouped by algorithm. Then, a Z hypothesis test was conducted to compare the average query execution time between QQESPM and QQ-simple. The calculated p-value of $7.929 \cdot 10^{-6}$ confirms a statistically significant difference in average execution times between the two algorithms when the dataset size is sufficiently large. 
\protect
\section{Conclusion}
\protect
The main objective of this study was to introduce and formally define a new category of POI group search called QQ-SPM, which generalizes the existing SPM query by including connectivity constraints among POIs. To address the proposed QQ-SPM search problem, we introduced the QQESPM algorithm, derived from ESPM. We conducted an empirical evaluation comparing the runtime performance of the QQESPM algorithm with a simplified QQ-SPM solution that only verifies connectivity constraints in the final stage of an ESPM execution. Additionally, we performed a statistical hypothesis test to assess the average runtime of QQESPM and the trivial solution. The experimental results, supported by statistical analyses, confirm the effectiveness of the QQESPM algorithm, highlighting its efficiency and superior performance in executing QQ-SPM queries. For future work, we plan to enrich the set of available spatial predicates for defining spatial search patterns. We will also conduct more extensive performance evaluation, using code parallelization approaches.
\protect
\bibliographystyle{sbc}
\bibliography{sbc-template}
\end{document}